\documentclass[10pt]{amsart} 

\usepackage{amsmath,amssymb,amsthm,amsfonts}
\usepackage{natbib}

\newcommand{\T}{^\mathrm{T}}

\DeclareMathOperator{\Min}{Min}
\DeclareMathOperator{\Max}{Max}
\DeclareMathOperator{\Prob}{P}
\DeclareMathOperator{\Ex}{E}
\DeclareMathOperator{\Var}{Var}
\DeclareMathOperator{\Exp}{exp}

\newtheorem{theorem}{Theorem}[section]
\newtheorem{corollary}[theorem]{Corollary}
\newtheorem{lemma}[theorem]{Lemma}
\newtheorem{proposition}[theorem]{Proposition}

\newtheorem{definition}[theorem]{Definition}

\begin{document}

\title{Permutation tests of non-exchangeable null models}

\author{Jeffrey ROACH}
\address{Research Computing Center\\
	University of North Carolina at Chapel Hill\\
	Chapel Hill, NC 27599}
\email{roachjm@email.unc.edu}
\author{William VALDAR}
\address{Department of Genetics\\
        University of North Carolina at Chapel Hill\\
	Chapel Hill, NC 27599}
\email{william.valdar@unc.edu}

\thanks{Research reported in this manuscript was supported by the National Institute of General Medical Sciences of the National Institutes of Health under award number R01 GM104125 to W. Valdar}

\begin{abstract}
Generalizations to the permutation test are introduced to allow for situations in which the null model is not exchangeable.  It is shown that the generalized permutation tests are exact, and a partial converse: that any test function that is exact on all probability densities coincides with a generalized permutation test on a particular region, is established.  A most powerful generalized permutation test is derived in closed form. Approximations to the most powerful generalized permutation test are proposed to reduce the computational burden required to compute the complete test. In particular, an explicit form for the approximate test is derived in terms of a multinomial Bernstein polynomial approximation, and its convergence to the most powerful generalized permutation test is demonstrated. In the case where the determination of $p$-values is of greater interest than testing of hypotheses, two approaches to estimation of significance are analyzed. Bounds on the deviation from significance of the exact most powerful test are given in terms of sample size. For both estimators, as sample size approaches infinity, the estimator converges to the significance of the most powerful generalized permutation test under mild conditions. Applications of generalized permutation testing to linear mixed models are provided.
\end{abstract}

\maketitle

\section{Introduction and Preliminaries}
\subsection{Introduction}

The permutation test is an intuitive and popular method for estimating the level of significance of a statistic under an exchangeable null model. The fact that exchangeability of the observations is the only constraint imposed under the null is a key source of its persistent appeal. Parametric alternatives such as the likelihood ratio test, by contrast, rely on specifying a sampling density that may be conditionally exchangeable but that also imposes probabilistic constraints on the conditional distribution of the observations. These distributional constraints are not always reasonable. Under some data settings, for example, exchangeable null densities that are analytically convenient may nonetheless imply an intolerance to outliers that is felt to be unrealistic and likely to produce significance tests that are anti-conservative.

The origins of the permutation methods are found in the work of \cite{Fisher} and \cite{Pitman} in formulations of exact tests of significance where the underlying model distribution is of unspecified form. It was demonstrated by \cite{Scheffe} that for a large class of problems, permutation methods of Fisher and Pitman are necessary conditions in defining exact tests of significance. At the time that permutation tests were introduced, the amount of computation required to perform these tests greatly limited their application, despite theoretical results suggesting their power and effectiveness. Results of \cite{Welch} and \cite{Hoeffding1} show that permutation methods can achieve power approaching that obtained in parametric approaches.  Furthermore results of \cite{WaldWolf} and \cite{Hoeffding2} establish limiting distributions of statistics derived from permutation methods.

With the advent and near ubiquity of high-speed digital computers, however, the amount of computation required has become far less of a concern. Indeed, the intuitive appeal of permutation tests, in particular the fact that only exchangeability need be specified as opposed to a complete parametric form, has led to applications of permutation methods when other methods may be more appropriate.

The specialization of permutation methods appropriate for specific conditions where exchangeability under the null is not a reasonable assumption have been developed by a number of authors.  In particular, \cite{Schmoyer} considers the case where permutations are applied to the residuals of certain linear models where it is shown that test statistics for these models are asymptotically normal.  Generalized linear mixed-models are considered in \cite{Fitzmaurice} and \cite{LeeBraun}. \cite{AnderRob} and \cite{AnderTerBraak} consider modifications to permutation methods for partial tests in multiple linear regression and multi-factorial ANOVA, respectively.  These works investigate both the restriction of permutations to exchangeable sub-groupings and application of permutations to units other than the raw observations, such as the residuals. 

In this communication the permutation test construction is generalized allowing for non-exchangeability in the null model. The first section concludes with mathematical preliminaries and a brief review of results in the case where the null model is assumed to be exchangeable.  The second section defines a generalized permutation test for a given non-exchangeable null model.  These generalized permutation tests reduce to the traditional permutation test in the specialized case where the null model is exchangeable.  It is shown that the generalized permutation tests are exact and conversely that any test that is exact for all probability distributions coincides in large part to a generalized permutation test. The section continues with a construction of the most powerful generalized permutation test by reducing to a canonical problem in linear optimization. Connections with most powerful tests characterized by the Neyman-Pearson lemma are discussed briefly.  In the following section generalized permutation test analysis applied to linear mixed models is considered.

As in the case of the permutation test for exchangeable null models, the amount of computation required to perform the test can be considerable.  Therefore approximate tests are of fundamental interest.  The third section considers approximation of the most powerful generalized permutation test, and relates multinomial generalizations of Bernstein polynomial approximations to a Monte Carlo approximation.  It is shown that this approximation converges to the most powerful generalized permutation test as the number of permutations sampled with replacement approaches infinity when the permutations are sampled uniformly.

In practice, formal hypothesis testing may be of less interest than significance testing, {\em i.e.} the estimation of $p$-values.  In the final section, the most powerful generalized permutation test is reformulated in terms of an estimate of $p$-values.  Again, approximation approaches are central to practical application.  Two schemes are given: a direct approach and an indirect approach.  In the direct approach permutations are sampled with replacement such that each permutation is equally likely to be selected.  Its contribution to the estimator, however, differs based on the relative likelihood of the null model for that permutation.  Bounds in terms of number of permutations sampled are given, and convergence to the $p$-value based on the complete generalized permutation test is demonstrated.

The indirect approach introduces a more analytically tractable estimator than that of the direct approach.  Here permutations are sampled with replacement non-uniformly.  The contribution to the estimator, however, depends only on whether the relative likelihood for that permutation is greater than that of the identify permutation.  Bounds in terms of the number of sampled permutations are established and convergence as the number of sample permutations increases to infinity is proven for the estimator derived from the indirect approach.

\subsection{Preliminaries}

To formalize the notion of applying a permutation to a vector or a vector function, let $\pi$ be some permutation taken from the group of all permutations $\mathfrak{S}_n$ applied to the coordinate axes of $\mathbb{R}^n$.  The effect of $\pi$ on the vector $\mathbf{x}$ is a reordering of the indices of the coordinates of $\mathbf{x}$, {\em i.e.}:
$$ \pi \mathbf{x} = \pi \left( x_1, x_2, \ldots, x_n\right) = \left( x_{\pi\left(1\right)},
 x_{\pi\left(2\right)}, \ldots, x_{\pi\left(n\right)}\right) $$
for a particular $\mathbf{x} \in \mathbb{R}^n$.

The permutation of coordinates induces a natural decomposition of the space $\mathbb{R}^n$ into a finite number of subsets.  This decomposition is central to the analysis of permutation tests. For a particular $\mathbf{x} \in \mathbb{R}^n$, consider the images of $\mathbf{x}$ under all permutations of coordinates.  If this set of images consists of $n!$ distinct values, then there exists exactly one permutation $\pi$ whose application will result in a vector whose coordinates are in strictly ascending order.  That is, there exists $\pi \in \mathfrak{S}_n$ where $\pi \mathbf{x}$ is such that: $x_{\pi\left(1\right)} < x_{\pi\left(2\right)} < \ldots < x_{\pi\left(n\right)}$.  Let $T$ be the set of points in $\mathbb{R}^n$ whose coordinates are in strictly ascending order, {\em i.e.}:
$$T = \left\{ \mathbf{x} \in \mathbb{R}^n : x_1 < x_2 < \ldots < x_n \right\}.$$
Images of $T$ under the action of distinct elements of $\mathfrak{S}_n$ are disjoint; that is,
$$ \pi_i T \cap \pi_j T = \emptyset \quad\text{for}\quad \pi_i \neq \pi_j \,.$$
Conversely let the set $E$ be the set of points in $\mathbb{R}^n$ where two or more of the coordinates are equal, {\em i.e.}:
$$E = \left\{ \mathbf{x} \in \mathbb{R}^n : x_i = x_j \text{ for some } i \neq j\right\}.$$ 
In terms of permutations, the set $E$ consists of those points in $\mathbb{R}^n$ that are fixed by two or more elements of $\mathfrak{S}_n$; that is,
$$E=\{\mathbf{x}\in\mathbb{R}^n: \mathbf{x}=\pi_i \mathbf{x} \quad\text{for}\quad \pi_i\ne 1\} \,.$$
Geometrically, the points of $E$ correspond to the edges between the images of $T$ under different permutations. Therefore, the space $\mathbb{R}^n$ is partitioned into a finite number of subsets:
\begin{equation}\label{eq:decomp}
 \mathbb{R}^n = \bigcup_{\pi \in \mathfrak{S}_n} \pi T \cup E
\end{equation}
where $\pi_i T \cap \pi_j T = \emptyset$ for $\pi_i \neq \pi_j$.  Note that the set $E$ has Lebesgue measure zero.

An arbitrary subset $\Omega \subseteq \mathbb{R}^n$ is said to be $\mathfrak{S}_n$-{\em symmetric} if it contains the image under all permutations of every point in the subset;  that is,
$$ \left\{ \pi \mathbf{x} : x \in \Omega \text{ and } \pi \in \mathfrak{S}_n\right\} \subseteq \Omega.$$
Note that, from this definition, $\mathbb{R}^n$ is $\mathfrak{S}_n$-symmetric.  Throughout this communication, we will assume the somewhat more general situation where the sample space, $\Omega$, is not necessarily $\mathbb{R}^n$ but some$\mathfrak{S}_n$-symmetric and Lebesgue-measurable subset.  Let $\Omega_T$ denote $T \cap \Omega$.

Permutation tests are defined to be test functions that satisfy a particular algebraic relationship between points in the sample space and the images of those points under permutation.  In the case where a test function is not randomized and therefore takes only the values $0$ or $1$, then a permutation test is a test function where for a given point, the test function will take the value $1$ for $\alpha$ of images of that point under permutation and will take the value $0$ for the remaining $1-\alpha$ images under permutation.  In the more general setting, allowing randomized test functions, a {\em permutation test} is defined as follows:
\begin{definition}\label{def:perm-test} For any $\alpha \in \left[0,1\right]$, a function $\varphi : \Omega \rightarrow \left[ 0,1\right]$ satisfying the property:
$$ \frac{1}{n!}\sum_{\pi \in \mathfrak{S}_n} \varphi\left(\pi \mathbf{x}\right) = \alpha \quad \text{a.e.}$$
is said to be a {\em permutation test} at significance level $\alpha$.
\end{definition}

Of particular interest in the study of permutation tests are functions that remain unchanged by the application of all permutations.  Such functions are said to be {\em exchangeable}:
\begin{definition} A function $f : \Omega \rightarrow \mathbb{R}$ is said to be {\em exchangeable} if $f\left(\pi \mathbf{x}\right) = f\left(\mathbf{x}\right)$ for all $\pi \in \mathfrak{S}_n$ and all $\mathbf{x} \in \Omega$.  That is, for all $\mathbf{x} \in \Omega$, $f$ takes the same value for all permutations of $\mathbf{x}$.
\end{definition}

Exchangeable functions enjoy a number of properties related to the invariance under permutation of coordinate axes.  In particular:

\begin{lemma}\label{lem:initiallemma} For $g_0 : \mathbb{R}^n \rightarrow \mathbb{R}$ exchangeable and Lebesgue-measurable:
\begin{itemize}
\item[a)] For any Lebesgue-measurable set $G \subseteq \mathbb{R}^n$ and permutation $\pi \in \mathfrak{S}_n$:
$$ \int_{\pi G} g_0\,d\mu = \int_{G} g_0\,d\mu.$$
\end{itemize}
Additionally for $\mathfrak{S}_n$-symmetric and Lebesgue measurable $\Omega$:
\begin{itemize}
\item[b)] For $g_0 : \Omega \rightarrow \mathbb{R}$:
$$ \int_\Omega g_0\,d\mu = n!\int_{\Omega_T} g_0\,d\mu. $$
\item[c)] For $\varphi$ a permutation test at significance level $\alpha$:
$$ \int_\Omega g_0\varphi\,d\mu = \alpha n! \int_{\Omega_T} g_0\,d\mu. $$
\end{itemize}
\end{lemma}
\begin{proof}
Statement a) follows from a change of variables applied to $g_0$ exchangeable:
$$ \int_{\pi G} g_0\,d\mu = \int_G g_0\left(\pi \mathbf{x}\right)\,d\mu\left(\mathbf{x}\right) = \int_G g_0\,d\mu $$
since the Jacobian of the linear operation corresponding to the permutation $\pi$ has absolute value one.

Statement b) follows from the application of a) to the decomposition (\ref{eq:decomp}):
$$ \int_\Omega g_0\,d\mu = \sum_{\pi \in \mathfrak{S}_n} \int_{\pi \Omega_T} g_0\,d\mu = n!
 \int_{\Omega_T} g_0\,d\mu .$$

Statement c) follows from b) and the exchangeability of $g_0$:
\begin{align*}
 \int_\Omega g_0 \varphi\,d\mu &= \sum_{\pi \in \mathfrak{S}_n} \int_{\pi \Omega_T} g_0 \varphi\,d\mu \\
      &= \sum_{\pi \in \mathfrak{S}_n} \int_{\Omega_T} g_0\left(\pi \mathbf{x}\right)
 \varphi\left(\pi \mathbf{x}\right)\,d\mu\left(\mathbf{x}\right)\\
      &= \int_{\Omega_T} g_0\left(\mathbf{x}\right) \sum_{\pi \in \mathfrak{S}_n}
 \varphi\left(\pi \mathbf{x}\right)\,d\mu\left(\mathbf{x}\right)\\
      &= \alpha n! \int_{\Omega_T} g_0\,d\mu
\end{align*}
as required.
\end{proof}

In the case where $g_0$ represents a probability density, then part c) of the theorem can be interpreted as the statement that the permutation test $\varphi$ is exact:

\begin{corollary}[Scheff\'e] If $\varphi$ is a permutation test at significance level $\alpha$ and $g_0$ is an exchangeable probability density, then
$$ \int_\Omega g_0 \varphi\,d\mu = \alpha$$
{\em i.e.} $\varphi$ is exact (\cite{Lehmann}, Theorem 5.8.1), (\cite{Scheffe}, Theorem 2).
\end{corollary}

\section{Generalized Permutation Test Functions}

\subsection{Definition of Generalized Permutation Test Functions and Exactness}

The exactness of permutation tests on exchangeable probability distributions provides the impetus for the definition of generalized permutation tests.  Ultimately it will be shown that a generalized permutation test is exact on arbitrary probability densities.  The definition of the generalized permutation test is given in terms of a particular exchangeable function, $h_{g_0}$.  Let $h_{g_0}$, or $h$ where it is clear from context, denote the exchangeable function:
$$  h_{g_0}\left(\mathbf{x}\right) =  \frac{1}{n!} \sum_{\pi \in \mathfrak{S}_n} g_0\left(\pi \mathbf{x}\right).$$
for $g_0$ an arbitrary real-valued function of $\Omega$.

Note that on symmetric sets, $g_0$ and $h$ satisfy the following elementary properties:
\begin{lemma}\label{lem:simplelemma} Let $\Omega$ be a symmetric, measurable set.  Then for $g_0 : \Omega \rightarrow \mathbb{R}$ positive and measurable:
\begin{itemize}
\item[a)] 
$$ \int_{\Omega} g_0\,d\mu = \int_{\Omega} h\,d\mu.$$
\item[b)] If there exists a subset of $\Omega$ of non-zero measure where $h > 0$, then there exists a subset of $\Omega$ of non-zero measure where $g_0>0$.
\end{itemize}
\end{lemma}

The definition of a {\em generalized permutation test function} with respect to any arbitrary positive real-valued function $g_0$ is then:
\begin{definition} For $\alpha \in \left[0,1\right]$, $g_0 : \Omega \rightarrow \mathbb{R}$, and $ h_{g_0}\left(\mathbf{x}\right) =  \frac{1}{n!} \sum_{\pi \in \mathfrak{S}_n} g_0\left(\pi \mathbf{x}\right)$ non-zero; a function $\varphi : \Omega \rightarrow \left[ 0,1\right]$ satisfying the property:
\begin{displaymath}
\frac{1}{n!}\displaystyle \sum_{\pi \in \mathfrak{S}_n} 
  \frac{g_0\left( \pi \mathbf{x} \right)}{h\left(\mathbf{x}\right)} \,\,\,
 \varphi\left(\pi \mathbf{x}\right) = \alpha \quad \text{a.e.}
\end{displaymath}
is said to be a  {\em generalized permutation test at significance level} $\alpha$ with respect to $g_0$.
\end{definition}

Lemma \ref{lem:initiallemma} for exchangeable, positive real-valued functions, has the following generalization in the non-exchangeable case:
\begin{lemma} For $g_0 : \Omega \rightarrow \mathbb{R}$ and $\varphi$ a generalized permutation test of significance $\alpha$:
$$ \int_\Omega g_0\varphi\,d\mu = \alpha n! \int_{\Omega_T} h\,d\mu $$
where  $h_{g_0}\left(\mathbf{x}\right) =  \frac{1}{n!} \sum_{\pi \in \mathfrak{S}_n} g_0\left(\pi \mathbf{x}\right)$ non-zero {\em a.e.}
\end{lemma}
\begin{proof}
Following the decomposition (\ref{eq:decomp}):
\begin{align*}
 \int_\Omega g_0 \varphi\,d\mu &= \sum_{\pi \in \mathfrak{S}_n} \int_{\pi \Omega_T} g_0 \varphi\,d\mu \\
      &= \sum_{\pi \in \mathfrak{S}_n} \int_{\Omega_T} g_0\left(\pi \mathbf{x}\right)
 \varphi\left(\pi \mathbf{x}\right)\,d\mu\left(\mathbf{x}\right).
\end{align*}
Introducing $h$:
$$ \int_\Omega g_0 \varphi\,d\mu = \sum_{\pi \in \mathfrak{S}_n} \int_{\Omega_T} h\left(\mathbf{x}\right)
								  \frac{g_0\left(\pi \mathbf{x}\right)}{h\left( \mathbf{x}\right)}
                                  \varphi\left(\pi \mathbf{x}\right)\,d\mu\left(\mathbf{x}\right). $$
By construction, $h$ is exchangeable, therefore:
\begin{align*}
\int_\Omega g_0 \varphi\,d\mu &= \int_{\Omega_T} h\left(\mathbf{x}\right) \sum_{\pi \in \mathfrak{S}_n} 
								  \frac{g_0\left(\pi \mathbf{x}\right)}{h\left( \mathbf{x}\right)}
                                  \varphi\left(\pi \mathbf{x}\right)\,d\mu\left(\mathbf{x}\right) \\
      &= \alpha n! \int_{\Omega_T} h\,d\mu
\end{align*}
completing the proof.
\end{proof}

As in the exchangeable case, when the function $g_0$ is taken to represent the probability density under the null hypothesis, the previous lemma can be interpreted in terms of the statement that generalized permutation tests are exact.
\begin{corollary} If $\varphi$ is a generalized permutation test of significance $\alpha$ and $g_0$ is a probability density with non-zero $h_{g_0}$ {\em a.e.}, then
$$ \int_\Omega g_0 \varphi\,d\mu = \alpha;$$
{\em i.e.} $\varphi$ is exact.
\end{corollary}

It is also interesting to note that in the case where the generalized permutation test is defined with respect to a different density, the difference in bias is bounded by the $L^1$-distance between the two densities.  In particular, where $g$ is only known approximately and the generalized permutation test is defined with respect to the approximation $\tilde{g}$, then the difference in bias is bounded by the accuracy of the approximation in terms of $L^1$-distance.
\begin{corollary} Let $g$ and $\tilde{g}$ be probability densities where $\tilde{h}\left(\mathbf{x}\right) =  \frac{1}{n!} \sum_{\pi \in \mathfrak{S}_n} \tilde{g}\left(\pi \mathbf{x}\right)$ is non-zero {\em a.e.}  If 
$$\frac{1}{n!}\displaystyle \sum_{\pi \in \mathfrak{S}_n} 
  \frac{\tilde{g}\left( \pi \mathbf{x} \right)}{\tilde{h}\left(\mathbf{x}\right)} \,\,\,
 \varphi\left(\pi \mathbf{x}\right) = \alpha \quad \text{a.e.},$$
then:
$$ \left| \int_\Omega g \varphi\,d\mu - \alpha \right| \leq \int_\Omega \left| g - \tilde{g} \right|\,d\mu; $$
{\em i.e.} the difference in bias is bounded by the $L^1$-distance on $g$ and $\tilde{g}$.
\end{corollary}

The next proposition shows that the definition of a generalized permutation test is natural in the sense that any test that is exact for all absolutely continuous probability distributions will be identical to a generalized permutation test on the region where $h$ is non-zero.

\begin{proposition} Let $\varphi : \Omega \rightarrow \left[0,1\right]$ be such that:
\begin{equation}\label{eq:exactness}
\int_\Omega g \varphi\,d\mu = \alpha
\end{equation}
for all probability densities $g$.  Then, for $h\left(\mathbf{x}\right) =  \frac{1}{n!} \sum_{\pi \in \mathfrak{S}_n} g_0\left(\pi \mathbf{x}\right)$:
$$ \frac{1}{n!}\displaystyle \sum_{\pi \in \mathfrak{S}_n} \frac{g\left( \pi \mathbf{x} \right)}{h\left(\mathbf{x}\right)} \,\,\,
 \varphi\left(\pi \mathbf{x}\right) = \alpha \quad \text{a.e.} $$
where $h\left(\mathbf{x}\right) \neq 0$.
\end{proposition}
\begin{proof}
For any given probability density $g$, let $\Omega^\prime$ be the subset of $\Omega$ where $h\left(\mathbf{x}\right) \neq 0$, {\em i.e.}:
$$ \Omega^\prime = \left\{ \mathbf{x} \in \Omega : h\left(\mathbf{x}\right) \neq 0 \right\}.$$
Consider two symmetric subsets of $\Omega^\prime$:
\begin{align*}
\Omega^\prime_+ &= \left\{ \mathbf{x} \in \Omega^\prime : \displaystyle \sum_{\pi \in \mathfrak{S}_n} \frac{g\left( \pi \mathbf{x} \right)}{h\left(\mathbf{x}\right)} \,\,\, \varphi\left(\pi \mathbf{x}\right) > \alpha n! \right\} \\
\Omega^\prime_- &= \left\{ \mathbf{x} \in \Omega^\prime : \displaystyle \sum_{\pi \in \mathfrak{S}_n} \frac{g\left( \pi \mathbf{x} \right)}{h\left(\mathbf{x}\right)} \,\,\, \varphi\left(\pi \mathbf{x}\right) < \alpha n! \right\}
\end{align*}
It will be shown that if a probability density $g$ exists where $\mu\left(\Omega^\prime_+\right) > 0$ or $\mu\left(\Omega^\prime_-\right) > 0$, then there exist probability densities $g_1$ and $g_2$ contradicting (\ref{eq:exactness}).

If $\mu\left(\Omega^\prime_+\right) > 0$, then
$$ \int_{\Omega^\prime_+} g\,d\mu > 0$$
by lemma \ref{lem:simplelemma}.  Therefore, where $1_{\Omega^\prime_+}$ is the indicator function of the set $\Omega^\prime_+$, $g_1$ defined by:
$$ g_1 = \frac{1}{\int_{\Omega^\prime_+}g\,d\mu} 1_{\Omega^\prime_+} g$$
is a probability density.  Note that $1_{\Omega^\prime_+}$ is exchangeable due to the fact that $\Omega^\prime_+$ is symmetric.  Consider $\int_\Omega g_1 \varphi\,d\mu$:
\begin{align*}
\int_\Omega g_1 \varphi\,d\mu &= \int_T \sum_{\pi \in \mathfrak{S}_n} g_1\left(\pi \mathbf{x}\right)\varphi\left(\pi \mathbf{x}\right)d\,\mu\left(\mathbf{x}\right) \\
&= \int_T \frac{1_{\Omega^\prime_+}\left(\mathbf{x}\right) h\left(\mathbf{x}\right)}{\int_{\Omega^\prime_+}g\,d\mu} \left[ \sum_{\pi \in \mathfrak{S}_n} \frac{g\left(\pi \mathbf{x}\right)}{h\left(\mathbf{x}\right)} \varphi\left(\pi \mathbf{x}\right) \right]d\,\mu\left(\mathbf{x}\right) \\
&> \alpha n! \frac{\int_{\Omega^\prime_+ \cap T} h\,d\mu}{\int_{\Omega^\prime_+}g\,d\mu}. 
\end{align*}
Therefore, by lemma \ref{lem:initiallemma} and lemma \ref{lem:simplelemma}, it follows from the observation that $\Omega^\prime_+$ is symmetric that
$$\int_{\Omega} g_1 \varphi\,d\mu > \alpha$$
contradicting (\ref{eq:exactness}).

An analogous contradiction can be derived using the probability density $g_2$ defined:
$$ g_2 = \frac{1}{\int_{\Omega^\prime_-}g\,d\mu} 1_{\Omega^\prime_-} g$$
in the case where $\mu\left(\Omega^\prime_-\right) > 0$.
\end{proof}

This result is a generalization to arbitrary probability densities of a result of Scheff\'e (\cite{Scheffe}, Theorem 3) concerning permutation tests applied to certain exchangeable probability distributions.

\subsection{Power of Generalized Permutation Test Functions}

The power of the generalized permutation test can be maximized as follows.  Following the decomposition of $\Omega$ with respect to the symmetric group of permutations $\mathfrak{S}_n$, note that for any positive function $g : \Omega \rightarrow \mathbb{R}$ and
$\varphi : \Omega \rightarrow \left[0,1\right]$,
$$ \int_\Omega g \varphi\,d\mu = \int_{\Omega_T} \left( \sum_{\pi \in \mathfrak{S}_n}
 g\left(\pi \mathbf{x}\right) \varphi \left(\pi \mathbf{x}\right) \right)\,d\mu\left(\mathbf{x}\right) .$$
Therefore for $\varphi_1$ and $\varphi_2$ generalized permutation test functions:
$$ \sum_{\pi \in \mathfrak{S}_n} g\left(\pi \mathbf{x}\right)\varphi_1\left(\pi \mathbf{x}\right) \geq 
\sum_{\pi \in \mathfrak{S}_n} g\left(\pi \mathbf{x}\right)\varphi_2\left(\pi \mathbf{x}\right) \text{ a.e.} $$
implies
$$ \int_\Omega g \varphi_1\,d\mu \geq \int_\Omega g \varphi_2\,d\mu .$$
Consequently given null and alternative hypotheses represented by positive functions $g_0$ and $g_1$ respectively, the most powerful generalized permutation test, $\varphi$, of significance $\alpha$ can be attained by selecting values for $\varphi\left(x\right)$, $x \in \Omega$, such that:
$$\sum_{\pi \in \mathfrak{S}_n} g_1\left(\pi \mathbf{x}\right)\varphi\left(\pi \mathbf{x}\right)$$ 
takes the maximum value subject to:
$$\sum_{\pi \in \mathfrak{S}_n} \frac{g_0\left(\pi \mathbf{x}\right)}{n!h\left(\mathbf{x}\right)}
\varphi\left(\pi \mathbf{x}\right) \leq \alpha$$
for $h\left(\mathbf{x}\right) =  \frac{1}{n!} \sum_{\pi \in \mathfrak{S}_n} g_0\left(\pi \mathbf{x}\right)$.  This optimization is exactly the {\em continuous knapsack problem} suggested by (\cite{Dantzig}, \S 26-1, p. 517; \cite{Martello}, \S 2.2, p. 16).  In the continuous knapsack problem, for given two vectors of dimension $n$, $\mathbf{v} \in \mathbb{R}^n$ and $\mathbf{w} \in \mathbb{R}^n$, the goal is to maximize $\mathbf{v}\T\mathbf{x}$ subject to $\mathbf{w}\T\mathbf{x} \leq \alpha$ for $\mathbf{x} \in \left[0,1\right]^n$.  A maximally powerful generalized permutation test is obtained by solving the continuous knapsack problem for set of permutation images of $\mathbf{x}$ for all $\mathbf{x} \in \Omega$.

Explicitly, given $g_0$ and $g_1$, define $l\left(\mathbf{x}\right)$ as follows:
$$ l\left(\mathbf{x}\right) = \left\{
                     \begin{array}{ll} \frac{g_1\left(\mathbf{x}\right)}{g_0\left(\mathbf{x}\right)}
                                                &\text{if } g_0\left(\mathbf{x}\right)\neq 0 \\
                                       \infty  &\text{if } g_0\left(\mathbf{x}\right)= 0 \text{ and }
                                                      g_1\left(\mathbf{x}\right)\neq 0 \\
                                              0 &\text{if } g_0\left(\mathbf{x}\right)= 0 \text{ and }
                                                      g_1\left(\mathbf{x}\right)= 0
             \end{array}
        \right.$$
for any $\mathbf{x} \in \Omega.$  For any given ordering of $\mathfrak{S}_n = \{\pi_1, \pi_2, \ldots, \pi_{n!}\}$, let  $l_i\left(\mathbf{x}\right) = l\left(\pi_i \mathbf{x}\right)$.  In particular, for a given fixed $\mathbf{x} \in \Omega$, let $\{l_1\left(\mathbf{x}\right), l_2\left(\mathbf{x}\right), \ldots, l_D\left(\mathbf{x}\right)\}$ be the distinct values of $l\left(\pi_i \mathbf{x}\right)$ for all $\pi_i \in \mathfrak{S}_n$.  Furthermore assume that $\mathfrak{S}_n$ is ordered such that:
\begin{equation}\label{eq:lr-chain}
l_1\left(\mathbf{x}\right) > l_2\left(\mathbf{x}\right) > l\left(\mathbf{x}\right) = l_r\left(\mathbf{x}\right) >
 \ldots > l_D\left(\mathbf{x}\right).
\end{equation}
Let $C_i = \{ \pi \in \mathfrak{S}_n : l_i\left(\mathbf{x}\right)\}$ be the classes of the partition of $\mathfrak{S}_n$ with respect to the ratio of $g_1$ to $g_0$ over all permutations of $\mathbf{x} \in \Omega$.  Define:
\begin{equation}\label{eq:weight}
w_i = \frac{\sum_{\pi \in C_i} g_0\left(\pi \mathbf{x}\right)}{n! h\left(\mathbf{x}\right)}
\end{equation}
where $h\left(\mathbf{x}\right) =  \frac{1}{n!} \sum_{\pi \in \mathfrak{S}_n} g_0\left(\pi \mathbf{x}\right)$.  Note that unlike the likelihood ratios which are decreasing by design, the sequence of $w_i$ is not necessarily monotonic.

Following the solution to the continuous knapsack problem, note that there exists a unique integer $d$ with $1 \leq d \leq D$ such that:
$$ \sum_{i=1}^{d-1} w_i < \alpha \leq \sum_{i=1}^d w_i . $$
Therefore the most powerful generalized permutation test is defined on the permutation images of $x \in \Omega$ by:
\begin{equation} \label{eq:alt-mp}
\varphi \left(\pi \mathbf{x}\right) = \left\{
             \begin{array}{ll} 1 &\text{if } \pi \in C_i \text{ for } i < d \\
                               \vartheta  &\text{if } \pi \in C_d \\
                               0 &\text{if } \pi \in C_i \text{ for } i > d
             \end{array}
        \right.
\end{equation}
where $\vartheta = \left(\alpha - \sum_{i=1}^{d-1} w_i \right)/w_d$.  \cite{LeeBraun} provide empirical support for tests of this form.

Recalling that $l\left(\mathbf{x}\right) = l_r\left(\mathbf{x}\right)$ by assumption, (\ref{eq:alt-mp}) can be written in a more convenient form.  Denote by $\Omega_r$ the subset of all $\mathbf{x} \in \Omega$ that rank $r$-th in the set of all permutation images with respect to the ratio of $g_1$ to $g_0$ {\em i.e.} satisfying (\ref{eq:lr-chain}). Then:

\begin{proposition}\label{prop:test-prop}
For $\mathbf{x} \in \Omega_r$, the most powerful generalized permutation test is given by:
\begin{displaymath}
\varphi \left(\mathbf{x}\right) = \left\{
             \begin{array}{ll} 1 &\text{if } \sum_{i=1}^r w_i < \alpha \\
                               \vartheta  &\text{if } \sum_{i=1}^{r-1} w_i < \alpha \leq \sum_{i=1}^r w_i \\
                               0 &\text{if } \sum_{i=1}^{r-1} w_i \geq \alpha
             \end{array}
        \right.
\end{displaymath}
where $\vartheta = \left(\alpha - \sum_{i=1}^{r-1} w_i \right)/w_r$.
\end{proposition}

\subsection{Most Powerful Generalized Permutation Tests and Most Powerful Tests}

In the case where $g_0$ is exchangeable, the most powerful generalized permutation test given by Proposition \ref{prop:test-prop}, specializes to the permutation test defined by \cite{LehmannStein1}.  It has been shown by \cite{LehmannStein1} that the permutation test is most powerful as a test of a non-exchangeable alternative hypothesis $g_1$ against a specific exchangeable null hypothesis constructed from $g_1$.  Lehmann and Stein observe that for a non-exchangeable hypothesis $g_1$, the permutation test reduces to a likelihood ratio test of the hypothesis represented by $g_1$ against the composite hypothesis defined by a particular one-parameter family of densities.  Following from a generalization of the Neyman-Pearson lemma proven in an earlier work of \cite{LehmannStein2}, a test is most powerful as a test of a hypothesis corresponding to a density $g_1$ against a hypothesis corresponding to a one-parameter family of densities if it is most powerful as a test of the hypothesis given by density $g_1$ against the hypothesis given by the density of the average over the parameter.  In the case of the permutation test, the average over the parameter in question, $\alpha$, corresponds to $h_{g_1}$.  Therefore, \cite{LehmannStein1} show that the permutation test is most powerful at significance level $\alpha$ as a test of the non-exchangeable alternative hypothesis $g_1$ against the exchangeable null hypothesis constructed from $h_{g_1}$.

Outside of the context considered in Lehmann and Stein, however, tests defined by Proposition \ref{prop:test-prop} are not necessarily most powerful against any null hypothesis, even where the null hypothesis is exchangeable.  Indeed, consider the following example.  Let $f\left(x,y\right)$ denote the bivariate normal distribution with mean at the origin and identity covariance.  Let the null hypothesis $g_0\left(x,y\right)$ be equal to $f\left(x,y\right)$.  Let the alternative hypothesis $g_1\left(x,y\right)$ be defined to be $f\left(x,y-\delta\right)$ for some $\delta > 0$.  Clearly the null hypothesis is exchangeable, and, for $\delta > 0$, the alternative hypothesis is not.  Following the Neyman-Pearson lemma \cite{Lehmann}, Theorem 3.2.1, the most powerful test at significance level $\frac{1}{2}$ rejects all points with $y>0$.  Any permutation test at significance level $\frac{1}{2}$, however, must satisfy Definition \ref{def:perm-test}.  Therefore the sum of the test function over all the permutation images must be one.  In this simple case, the permutation images correspond to pairs of points reflected through the $x=y$ line.  Therefore, the most powerful test, where the sum of the test function over the permutation images is either zero or two, cannot be achieved with a permutation test. 

A criterion for when the most powerful generalized permutation test is, in fact, most powerful is derived from the Neyman-Pearson lemma.  Let $D\left(\mathbf{x}\right)$ denote the number of distinct likelihood ratio $g_1\left(\pi \mathbf{x}\right) / g_0\left(\pi \mathbf{x}\right)$ values in set of permutation images of any given $\mathbf{x} \in \Omega$. Within each permutation image set there exists a $d\left(\mathbf{x}\right)$ with  $1 \leq d\left(\mathbf{x}\right) \leq D\left(\mathbf{x}\right)$ such that:
$$ \sum_{i=1}^{d\left(\mathbf{x}\right)-1} w_i < \alpha \leq \sum_{i=1}^{d\left(\mathbf{x}\right)} w_i $$
where the ordering of $\mathfrak{S}_n$ follows (\ref{eq:lr-chain}) and $w_i$ is defined following (\ref{eq:weight}) for any $\mathbf{x} \in \Omega$.  Let $l_d\left(\mathbf{x}\right)$ be the value of the likelihood ratio associated to the $d\left(\mathbf{x}\right)$-th permutation.  Then:
\begin{proposition}
The most powerful generalized permutation test is most powerful if and only if $l_d\left(\mathbf{x}\right)$ is constant
for almost all $\mathbf{x} \in \Omega$.
\end{proposition}

\section{Application of Generalized Permutation Tests to Linear Models}

The results of the previous section provide the initial elements of a somewhat more general theory of permutation testing; however, it is unrealistic to assume that the densities of $\mathbf{x}$ under null and alternative hypotheses are fully specified.  In fact, if they were, then from the results of the previous section, a more direct course of action would be through the application of the Neyman-Pearson theory.  However, the fundamental observation that only the order of the likelihood ratios obtained from the permutations of the observed sample need be specified as opposed to their exact values provides the basis for the next section.  As such, the principle difficulty in applying the generalized permutation test is determining the order by which the likelihood ratios of the permutations of the observed samples should fall.

Consider the linear model:
\begin{equation}\label{eq:lin-mod}
  \mathbf{y} = \mathbf{X}\boldsymbol{\beta} + \boldsymbol{\varepsilon}
\end{equation}
where $\mathbf{y} \in \mathbb{R}^n$, $\mathbf{X} \in \mathbb{R}^{n\times 1}$, $\boldsymbol{\beta} \in \mathbb{R}^p$, and $\mathbf{e} \in \mathbb{R}^{n}$.  The vector of observations $\mathbf{y}$ and the vector $\mathbf{X}$ are given, the value $\beta$ is unknown and the random vector $e$ is assumed to be taken from a multivariate normal distribution with mean at the origin and covariance given by a symmetric positive definite matrix.  Initially we will assume that this symmetric positive definite matrix is known under both the null and alternative hypothesis.

Let the null hypothesis, $H_0$, be that $\boldsymbol{\beta} = 0$.  Note that in this case, the density of $\mathbf{y}$ is not necessarily exchangeable; however, with known $\boldsymbol{\Sigma}_0$, the relative weights $w_i$ can be calculated.  A particular form of the alternative hypothesis that $\boldsymbol{\beta} \neq 0$ will be considered where $\mathbf{y}$ is multivariate normally distributed with mean $\mathbf{X}\boldsymbol{\beta}$ and known covariance $\boldsymbol{\Sigma}_1$.  Consider an alternative hypothesis where the unit vector $\mathbf{u} = \boldsymbol{\beta} / ||\boldsymbol{\beta}||$ is specified, but the magnitude $||\boldsymbol{\beta}||$ is unknown.  Note that this alternative hypothesis is a modest generalization of either the $p=1$ case where $\beta \in \mathbb{R}$ and $\beta > 0$ or, in the case of $p>1$, the alternative hypothesis that for some unique $j$, $\beta_j > 0$ and $\beta_i = 0$ for all $i\neq j$.

The likelihood ratio at any given $\mathbf{y}$ is therefore:
\begin{equation*}
  l\left(\mathbf{y}\right)  = \frac{f_1\left(\mathbf{y}\right)}{f_0\left(\mathbf{y}\right)}
  = \sqrt{\frac{\left|\boldsymbol{\Sigma}_0\right|}{\left|\boldsymbol{\Sigma}_1\right|}}
    \exp{\left(-\frac{1}{2}\left(\mathbf{y}-\mathbf{X}\boldsymbol{\beta}\right)^T\boldsymbol{\Sigma}_1^{-1}
    \left(\mathbf{y}-\mathbf{X}\boldsymbol{\beta}\right) + \mathbf{y}^T\boldsymbol{\Sigma}_0^{-1}\mathbf{y}\right)}
\end{equation*}
Or equivalently:
\begin{equation*}
l\left(\mathbf{y}\right) = S\left(\mathbf{X},\boldsymbol{\beta},\boldsymbol{\Sigma}_0,\boldsymbol{\Sigma}_1\right)
      \exp{\left(\boldsymbol{\beta}^T \mathbf{X}^T \boldsymbol{\Sigma}_1^{-1}\mathbf{y}\right)}
      \exp{\left(-\frac{1}{2}\mathbf{y}^T\left(\boldsymbol{\Sigma}_1^{-1}-\boldsymbol{\Sigma}_0^{-1}\right)\mathbf{y}\right)}
\end{equation*}
where $S\left(\mathbf{X},\boldsymbol{\beta},\boldsymbol{\Sigma}_0,\boldsymbol{\Sigma}_1\right)$ is given by:
\begin{equation*}
  S\left(\mathbf{X},\boldsymbol{\beta},\boldsymbol{\Sigma}_0,\boldsymbol{\Sigma}_1\right) =
  \sqrt{\frac{\left|\boldsymbol{\Sigma}_0\right|}{\left|\boldsymbol{\Sigma}_1\right|}}
  \exp{\left(-\frac{1}{2}\boldsymbol{\beta}^T\mathbf{X}^T\boldsymbol{\Sigma}_1^{-1}\mathbf{X}\boldsymbol{\beta}\right)}.
\end{equation*}
Introducing the generalized least squares estimate for $\boldsymbol{\beta}$, $\hat{\boldsymbol{\beta}}\left(\mathbf{y},\mathbf{X},\boldsymbol{\Sigma}_1\right)$:
\begin{equation*}
  \hat{\boldsymbol{\beta}}\left(\mathbf{y},\mathbf{X},\boldsymbol{\Sigma}_1\right)=\left(\mathbf{X}^T\boldsymbol{\Sigma}_1^{-1}\mathbf{X}\right)^{-1}\mathbf{X}^T\boldsymbol{\Sigma}_1^{-1}\mathbf{y}
\end{equation*}
results in:
\begin{equation*}
  l\left(\mathbf{y}\right) = S\left(\mathbf{X},\boldsymbol{\beta},\boldsymbol{\Sigma}_0,\boldsymbol{\Sigma}_1\right)
  \exp{\left(\left|\left|\boldsymbol{\beta}\right|\right| V_1\left(\mathbf{y},\mathbf{X},\boldsymbol{\Sigma}_0,\boldsymbol{\Sigma}_1,\mathbf{u}\right)\right)}
  V_2\left(\mathbf{y},\boldsymbol{\Sigma}_0,\boldsymbol{\Sigma_1}\right)
\end{equation*}
where:
\begin{align*}
  V_1\left(\mathbf{y},\mathbf{X},\boldsymbol{\Sigma}_0,\boldsymbol{\Sigma}_1,\mathbf{u}\right)&=\mathbf{u}^T \left(\mathbf{X}^T \boldsymbol{\Sigma}_1^{-1}\mathbf{X}\right)\hat{\boldsymbol{\beta}}\left(\mathbf{y},\mathbf{X},\boldsymbol{\Sigma}_1\right) \quad \text{and}\\
V_2\left(\mathbf{y},\boldsymbol{\Sigma}_0,\boldsymbol{\Sigma_1}\right)&=\exp{\left(-\frac{1}{2}\mathbf{y}^T\left[\boldsymbol{\Sigma}_1^{-1}-\boldsymbol{\Sigma}_0^{-1}\right]\mathbf{y}\right)}.
\end{align*}
Note that the function $S\left(\mathbf{X},\boldsymbol{\beta},\boldsymbol{\Sigma}_0,\boldsymbol{\Sigma}_1\right)$ is not a function of $\mathbf{y}$ and is therefore invariant of the action of $\pi \in \mathfrak{S}_n$.  Furthermore,
both $V_1\left(\mathbf{y},\mathbf{X},\boldsymbol{\Sigma}_0,\boldsymbol{\Sigma}_1,\mathbf{u}\right)$ and $V_2\left(\mathbf{y},\boldsymbol{\Sigma}_0,\boldsymbol{\Sigma_1}\right)$ can be calculated from known or assumed values.

Two possible tests can be considered.  If the covariance matrices under the null and alternative hypotheses can be assumed to differ only slightly, or in fact be identical, then the order of the likelihood ratios over the permutations of the observed sample will be determined by the order of the values of $V_1\left(\pi\mathbf{y},\mathbf{X},\boldsymbol{\Sigma}_0,\boldsymbol{\Sigma}_1,\mathbf{u}\right)$.  Alternatively, if the difference in covariance matrices can be assumed to dominate the likelihood ratios over the permutations of the observed sample, then the order of the values $V_2\left(\pi\mathbf{y},\boldsymbol{\Sigma}_0,\boldsymbol{\Sigma_1}\right)$ can be assumed to determined the order of the likelihood ratios.

Generalizing to the case where the covariance matrix $\boldsymbol{\Sigma}_1$ is completely unspecified is limited by the observation that permutations can not used to distinguish between likelihoods of different model specifications.  That is to say, if $\mathbf{y}$ is distributed following a multivariate normal distribution centered at $\mathbf{X}\boldsymbol{\beta}$ with covariance $\boldsymbol{\Sigma}$, them $\pi \mathbf{y}$ will be distributed following a multivariate normal distribution centered at $\boldsymbol{\Pi}\mathbf{X}\boldsymbol{\beta}$ with covariance $\boldsymbol{\Pi}^T\boldsymbol{\Sigma}\boldsymbol{\Pi}$ where $\boldsymbol{\Pi}$ is the usual matrix representation of the permutation $\pi$.

In the case where a reasonable (non-symmetric) priors for $\boldsymbol{\beta}$ and $\boldsymbol{\Sigma}$ are available, then the marginal distributions of $\mathbf{y}$, (\cite{OHagan}, pages 188, 244):
\begin{equation*}
m_i\left(\mathbf{y}\right) = \int g_i\left(\mathbf{y}|\boldsymbol{\beta}_i,\boldsymbol{\Sigma}_i\right)g_i\left(\boldsymbol{\beta}_i,\boldsymbol{\Sigma}_i\right)\,d\boldsymbol{\beta}_i\,d\boldsymbol{\Sigma}_i
\end{equation*}
where $i=0,1$ designates likelihood under null and alternative hypotheses respectively, can be calculated for each permutation of the observed sample.  The chain of decreasing likelihood ratios (\ref{eq:lr-chain}) and weights (\ref{eq:weight}) can be calculated and a permutation test can be constructed according to Proposition \ref{prop:test-prop}.  Note that in this case, the chain of likelihood ratios (\ref{eq:lr-chain}) corresponds to the values of the Bayes factors.

In the case where no suitable priors are available, a reasonable approach is to restrict the capacity of the linear model (\ref{eq:lin-mod}).  Consider the case where $\boldsymbol{\Sigma}$ takes the form $\sigma^2\left(\mathbf{A}+\lambda^2\mathbf{I}\right)$ for some know positive symmetric definite matrix $\mathbf{A}$ (\emph{c.f.} \cite{VisscherGoddard}).  If $\sigma_0^2$ and $\lambda_0^2$, the values of $\sigma^2$ and $\lambda^2$ under the null hypothesis are assumed, then provided that some process exists to approximate $\sigma_1^2$ and $\lambda_1^2$, the values of $\sigma^2$ and $\lambda^2$ under the alternative hypothesis, sufficiently closely that the ordering of the likelihood ratios over the permutations of the observed sample (\ref{eq:lr-chain}) is preserved, then this approximation to $\sigma_1^2$ and $\lambda_1^2$ can be used in place of the unknown values of $\sigma_1^2$ and $\lambda_1^2$ in the construction of the generalized permutation test.

In the restricted context, for certain hypotheses, the generalized permutation test can be constructed directly.  Consider the case where $\mathbf{X}=0$ and $\boldsymbol{\Sigma}$ takes the form $\sigma^2\left(\mathbf{A}+\lambda^2\mathbf{I}\right)$ as above.  Assume that $\lambda^2$ is known and identical under null and alternative hypotheses.  If $\sigma_0^2$ and $\sigma_1^2$ denote the value of $\sigma^2$ under null and alternative hypotheses respectively, then the likelihood ratio takes the form:
\begin{equation*}
  l\left(\pi \mathbf{y}\right) = \left(\frac{\sigma_0}{\sigma_1}\right)^n \exp{\left(\frac{\sigma_1^2-\sigma_0^2}{2 \sigma_0^2 \sigma_1^2} \left(\pi \mathbf{y}\right)^T \left(\mathbf{A}+\lambda^2 \mathbf{I}\right)^{-1} \left(\pi \mathbf{y}\right)\right)}.
\end{equation*}
Therefore, if the null hypothesis is the $\sigma^2$ takes some nominal value $\sigma_0^2$ and the alternative hypothesis is that $\sigma^2 > \sigma_0^2$, then the weights (\ref{eq:weight}) can be calculated from $\sigma_0^2$ and the ordering of the likelihood ratios (\ref{eq:lr-chain}) under the permutations of the observed samples can be determined from the ordering of $\left(\pi \mathbf{y}\right)^T \left(\mathbf{A}+\lambda^2 \mathbf{I}\right)^{-1} \left(\pi \mathbf{y}\right)$.  Conversely, the alternative hypothesis that $\sigma^2 < \sigma_0^2$ can be tested using the generalized permutation test with the reverse ordering. 

\section{Approximate Generalized Permutation Tests}

In the application of the generalized permutation test, the need to calculate $n!$ likelihood ratios poses a substantial practical obstacle. This section focuses on the mathematical structure of approximate methods that sample with replacement from the complete group of permutations $\mathfrak{S}_n$.  The theory is developed with respect to samples derived from an arbitrary discrete probability distribution on the group of permutations.  It is shown, in the case where each permutation, is equally likely to be selected for the sample, that the approximation converges to the most powerful generalized permutation test of the previous section.  Since the group of permutations is finite, albeit large, asymptotic results are of primarily of theoretical interest.  Nonetheless the bounds on the $p$-values estimated with respect to finite samples in the subsequent section are based on the mathematical structure developed in this section. 

A direct Monte Carlo approximation approach to the generalized permutation test is obtained by considering a sample with replacement from $\mathfrak{S}_n$ instead of the entire symmetric group.  Let $S$ be sample with replacement from $\mathfrak{S}_n$ of size $s$ where the probability of selecting a particular permutation $\pi$ is given by $p_{\pi}$.  Analogous to (\ref{eq:lr-chain}), for a given fixed $x \in \Omega$, let there be ${\hat D}$ distinct values of $l_i\left(\mathbf{x}\right) = g_1\left(\pi_i \mathbf{x}\right) / g_0\left(\pi_i \mathbf{x}\right)$ for all $\pi_i \in S$. Again assume that the permutations in $S$ are ordered such that:
\begin{equation}\label{eq:a-lr-chain}
l_1\left(\mathbf{x}\right) > l_2\left(\mathbf{x}\right) > l\left(\mathbf{x}\right) = l_{\hat r}\left(\mathbf{x}\right)
 > \ldots > l_{\hat D}\left(\mathbf{x}\right).
\end{equation}
Let ${\hat C}_i = \{ \pi \in S : g_1\left(\pi_i \mathbf{x}\right) / g_0\left(\pi_i \mathbf{x}\right) = l_i\left(\mathbf{x}\right)\}$ be the classes of the partition of $S$ with respect to the ratio of $g_1$ to $g_0$.

For fixed $S$, a suitable approximation to $\varphi$ is obtained by redefining the null hypothesis weights (\ref{eq:weight}) with respect to $S$.  Let the integer vector $\mathbf{k}=\left(k_1,k_2,\ldots,k_{\hat D}\right)$ represent the number of permutations, $k_i$, from ${\hat C}_i$ in $S$.  Note that no value $k_i$ is equal to zero. Define:
$$ {\hat w}_i = \frac{k_i\sum_{\pi \in {\hat C}_i} g_0\left(\pi \mathbf{x}\right)}{\sum_{j=1}^{\hat D} k_j
 \sum_{\pi \in {\hat C}_j} g_0\left(\pi \mathbf{x}\right)}. $$
Then ${\hat \varphi}_S$ is defined:
\begin{equation}\label{eq:a-eg}
{\hat \varphi}_S \left(\mathbf{x}\right) = \left\{
             \begin{array}{ll} 1 &\text{if } \sum_{i=1}^{\hat r} {\hat w}_i < \alpha \\
                               {\hat \vartheta}  &\text{if } \sum_{i=1}^{{\hat r}-1} {\hat w}_i <
                               \alpha \leq \sum_{i=1}^{\hat r} {\hat w}_i \\
                               0 &\text{if } \sum_{i=1}^{{\hat r}-1} {\hat w}_i \geq \alpha
             \end{array}
        \right.
\end{equation}
where ${\hat \vartheta} = \left(\alpha - \sum_{i=1}^{{\hat r}-1} {\hat w}_i \right)/{\hat w}_{\hat r}$ for each $\mathbf{x}$ satisfying (\ref{eq:a-lr-chain}).

Note that ${\hat \varphi}_S$ is defined for a specific sample with replacement $S$ and therefore does not take into account the stochastic aspect of sampling from the set of possible permutations.  Following the approach of \cite{Dwass} in the exchangeable case, define the randomized test function ${\hat \varphi}$ to be the probability that a sample with replacement $S$ is selected such that $x$ is rejected by the test function, ${\hat \varphi_S}$, defined according to (\ref{eq:a-eg}).  Therefore the test function ${\hat \varphi}$ corresponds to the Monte Carlo process of selecting a sample of permutation from the set of all possible permutations, constructing an appropriate ${\hat \varphi}_S$ from those permutations, and rejecting based on the result of that test. 

Let the vector $\mathbf{p}$ represent that probability $p_i$ of selecting a permutation in $C_i$ for each $C_i = \{ \pi \in \mathfrak{S}_n : g_1\left(\pi_i \mathbf{x}\right) / g_0\left(\pi_i \mathbf{x}\right) = l_i\left(\mathbf{x}\right)\} \subseteq \mathfrak{S}_n$.  For a sample $S$ with replacement from $\mathfrak{S}_n$ of size $s$, let the integer vector $\mathbf{k}=\left(k_1,k_2,\ldots,k_D\right)$ represent the number of permutations, $k_i$, from $C_i$ in $S$.  Note that the value $k_i$ is equal to zero if no permutation in the sample is taken from class $C_i$.

Consider the values obtained from $\mathbf{x}$ under permutation for $\mathbf{x} \in \Omega_r$.  Assume that the permutations of $\mathfrak{S}_n$ are ordered with respect to the $D_x$ distinct values of the likelihood ratio satisfying (\ref{eq:lr-chain}).  Any given sample with replacement $S$, represented by the vector $\mathbf{k}$, determines values:
\begin{equation}\label{eq:a-weight}
{\hat w}_{i,\mathbf{k}} = \frac{k_i\sum_{\pi \in C_i} g_0\left(\pi \mathbf{x}\right)}{\sum_{j=1}^{D_x} k_j
 \sum_{\pi \in C_j} g_0\left(\pi \mathbf{x}\right)}.
\end{equation}
Let $\Delta_S$ be given by:
$$\Delta_s=\left\{\left(k_1,k_2,\ldots,k_{D_x}\right) \in \mathbb{Z}^{D_x} : k_i\geq0 \text{ and }
 \sum_{i=1}^{D_x} k_i = s \right\}. $$
Within $\Delta_s$ define two subsets each depending on $r$:
\begin{align*}
&H_r = \left\{ \mathbf{k}\in \Delta_s : \sum_{i=1}^r {\hat w}_{i,\mathbf{k}} < \alpha \right\} \text{ and}\\
&B_r = \left\{ \mathbf{k}\in \Delta_s : \sum_{i=1}^{r-1} {\hat w}_{i,\mathbf{k}} < \alpha \leq 
\sum_{i=1}^r {\hat w}_{i,\mathbf{k}} \right\}
\end{align*}
for a given $\alpha$.

For $\mathbf{x} \in \Omega_r$, the value of ${\hat \varphi}\left(\mathbf{x}\right)$, the approximate most powerful generalized permutation test, is given by the probability that $\mathbf{x}$ will be rejected under a test function of the form (\ref{eq:a-eg}) corresponding to a randomly selected sample with replacement $S$.  Explicitly:
\begin{proposition}
For $\mathbf{x} \in \Omega_r$, sample size $s$, and probability of selecting each permutation class denoted by the probability vector $\mathbf{p}$, the approximate generalized permutation test of significance $\alpha$ is given by:
\begin{align*}
{\hat \varphi}\left(\mathbf{x};\mathbf{p}\right) = 
&\sum_{\mathbf{k}\in H_r} \binom{s}{k_1, k_2, \ldots k_D} p_1^{k_1} p_2^{k_2} \ldots p_D^{k_D} + \\
&\sum_{\mathbf{k}\in B_r} {\hat \vartheta}_{\mathbf{k}} \binom{s}{k_1, k_2, \ldots k_D}
 p_1^{k_1} p_2^{k_2} \ldots p_D^{k_D}
\end{align*}
where ${\hat \vartheta}_{\mathbf{k}} = \left(\alpha - \sum_{i=1}^{r-1} {\hat w}_{i,\mathbf{k}}
 \right)/{\hat w}_{r,\mathbf{k}}$.
\end{proposition}

Written in the above form, ${\hat \varphi}$ is seen to be a multinomial generalization of a Bernstein polynomial approximation (\cite{Lorentz}, \S 2.9, page 51).  Fixing $\mathbf{x}$, let $\Delta^*$ represent the $\left(D_x-1\right)$-th dimensional simplex:
$$\Delta^*=\left\{\left(p_1,p_2,\ldots,p_{D_x}\right) \in \mathbb{R}^{D_x} \left\vert\right. p_i\geq0 \text{ and }
 \sum_{i=1}^{D_x} p_i = 1 \right\}.$$
For points in $\Delta^*$, it is convenient to work in a modified $L^1$-norm:
\begin{equation}\label{eq:l1norm}
\|\mathbf{p}\|_{\mathbf{x}} = \sum_{i=1}^{D_x} |p_i| \sum_{\pi \in C_i} g_0\left(\pi \mathbf{x}\right)
\end{equation}
and the associated metric $\|\mathbf{p}-\mathbf{q}\|_{\mathbf{x}}$ for $\mathbf{p},\mathbf{q} \in \Delta^*$.  Weights analogous to (\ref{eq:a-weight}) defined based on the values of the function $g_0$ for all the permutations of $\mathbf{x}$:
\begin{equation*}
w^*_i = \frac{p_i\sum_{\pi \in C_i} g_0\left(\pi \mathbf{x}\right)}{\|\mathbf{p}\|_{\mathbf{x}}} =
 \frac{p_i\sum_{\pi \in C_i} g_0\left(\pi \mathbf{x}\right)}{\sum_{j=1}^D p_j
 \sum_{\pi \in C_j} g_0\left(\pi \mathbf{x}\right)}.
\end{equation*}
With respect to these weights, two subsets, analogous to $H_r$ and $B_r$, each depending on $r$ are defined within $\Delta^*$:
\begin{align*}
&H^*_r = \left\{ \mathbf{p}\in \Delta^* : \sum_{i=1}^r w^*_i < \alpha \right\} \text{ and}\\
&B^*_r = \left\{ \mathbf{p}\in \Delta^* : \sum_{i=1}^{r-1} w^*_i < \alpha \leq \sum_{i=1}^r w^*_i \right\}
\end{align*}
for a given $\alpha$.  Note that $\frac{1}{s}\mathbf{k} \in H_r^*$ if and only if $\mathbf{k} \in H_r$ and analogously $\frac{1}{s}\mathbf{k} \in B_r^*$ if and only if $\mathbf{k} \in B_r$.

For $\mathbf{k} \in \Delta_s$ and $\mathbf{p} \in \Delta^*$, define the multinomial kernel  $M_{\mathbf{k}}\left(\mathbf{p}\right)$ as follows:
$$ M_{\mathbf{k}}\left(\mathbf{p}\right) = \binom{s}{k_1, k_2, \ldots k_{D_x}} p_1^{k_1} p_2^{k_2}
 \ldots p_{D_x}^{k_{D_x}}.$$
The following concentration of measure property of the multinomial kernel is a variation on the Law of Large Numbers for multinomial distributions:
\begin{lemma} For arbitrary $\mathbf{p} \in \Delta^*$, define $\Delta_+$:
\begin{align*}
\Delta_+ &= \left\{ \mathbf{k} \in \Delta_s : \left\|\mathbf{p} - \frac{1}{s}\mathbf{k}\right\|_{\mathbf{x}} \geq
 \delta \right\} \\
&=\left\{ \mathbf{k} \in \Delta_s : \sum_{i=1}^{D_x} \left(p_i - \frac{k_i}{s}\right)
 \sum_{C_i} g_0 \left(\mathbf{x}\right) \geq \delta \right\}.
\end{align*}
There exist constants $c$ and $\sigma$ such that:
\begin{align*}
&\sum_{\mathbf{k} \in \Delta_+} M_{\mathbf{k}}\left(\mathbf{p}\right) \leq
 2 \Exp{\left(-\frac{1}{2}\frac{s \delta^2}{\sigma^2 + c \delta/3}\right)} \text{ and}\\
&\sum_{\mathbf{k} \in \Delta_+} M_{\mathbf{k}}\left(\mathbf{p}\right) \leq
 2 \Exp{\left(-\frac{s \delta^2}{2 c}\right)}
\end{align*}
for all $\delta > 0$.
\end{lemma}
\begin{proof}
Let $U$ be a discrete-valued random variable where:
$$ P\left[ U = \sum_{\pi \in C_i} g_0\left(\pi \mathbf{x}\right) \right] = p_i \text{ for } j=1,\ldots,D_x $$
and zero for all other values.  Consider $U_i$, $s$ independent random variables distributed identically to $U$.  The result follows from the application of the inequalities of Bernstein \cite{Uspensky} and Hoeffding \cite{Hoeffding3} respectively to the random variable defined by the mean of the $U_i$. Therefore the constant $\sigma^2$ is equal to $\Var{U}$ and $c$ is such that
$\left| U - \Ex{U} \right| \leq c$.
\end{proof}

For any function $f:\Delta^* \rightarrow \mathbb{R}$, the multinomial generalization of the Bernstein polynomial approximation is given by:
$$B^f_s\left(\mathbf{p}\right) = \sum_{\mathbf{k} \in \Delta_s}
f\left(\frac{k_1}{s},\frac{k_2}{s},\ldots\frac{k_{D_x}}{s}\right)
 \binom{s}{k_1, k_2, \ldots k_{D_x}} p_1^{k_1} p_2^{k_2} \ldots p_{D_x}^{k_{D_x}}.$$
Therefore, ${\hat \varphi}\left(\mathbf{x};\mathbf{p}\right)$, as a function of $\mathbf{p}$, is the multinomial Bernstein polynomial approximation to the function $\chi : \Delta^* \rightarrow \mathbb{R}$ defined:
\begin{align*}
\chi \left(\mathbf{p}\right) = \left\{
             \begin{array}{ll} 1 &\text{if } \mathbf{p} \in H^*_r \\
                               \vartheta  &\text{if } \mathbf{p} \in B^*_r \\
                               0 &\text{if } \mathbf{p} \not\in H^*_r \cup B^*_r 
             \end{array}
        \right.
\end{align*}
where $\vartheta = \left(\alpha - \sum_{i=1}^{r-1} w^*_i\right)/w^*_r$.  Analogous to traditional Bernstein polynomial approximations, convergence at a point $\mathbf{p} \in \Delta^*$ as $s$ approaches infinity depends on the continuity of $\chi$ at $\mathbf{p}$. 

\begin{lemma} For points $\mathbf{x} \in \Omega_r$, the function $\chi$, defined above, is continuous as a function of $\Delta^*$ in the modified $L^1$-norm metric (\ref{eq:l1norm}) at points $\mathbf{p}$ where $p_r$ is non-zero.
\end{lemma}
\begin{proof}
Note that for $\mathbf{p}$, $\mathbf{q} \in \Delta^*$:
\begin{align*}
\left|\sum_{i=j_1}^{j_2}\right. &\left.w_i^*\left(\mathbf{p}\right) -
\sum_{i=j_i}^{j_2} w_i^*\left(\mathbf{q}\right)\right|\\
&= \left| \frac{\sum_{i=j_1}^{j_2} p_i \sum_{\pi \in C_i} g_0\left(\mathbf{x}\right)}{\left\|\mathbf{p}\right\|_{\mathbf{x}}} -
         \frac{\sum_{i=j_1}^{j_2} q_i \sum_{\pi \in C_i} g_0\left(\mathbf{x}\right)}{\left\|\mathbf{p}\right\|_{\mathbf{x}}} + \right.\\ 
&\quad\quad\quad\quad\quad\quad\quad\,\,\, \left. \frac{\sum_{i=j_1}^{j_2} q_i \sum_{\pi \in C_i}
 g_0\left(\mathbf{x}\right)}{\left\|\mathbf{p}\right\|_{\mathbf{x}}} -
         \frac{\sum_{i=j_1}^{j_2} q_i \sum_{\pi \in C_i} g_0\left(\mathbf{x}\right)}{\left\|\mathbf{q}\right\|_{\mathbf{x}}} \right|\\
&= \frac{1}{\left\|\mathbf{p}\right\|_{\mathbf{x}}}\left| \sum_{i=j_1}^{j_2} \left( p_i - q_i \right) \sum_{\pi \in C_i}
 g_0\left(\mathbf{x}\right) + \frac{\sum_{i=j_1}^{j_2} q_i \sum_{\pi \in C_i}
 g_0\left(\mathbf{x}\right)}{\left\|\mathbf{q}\right\|_{\mathbf{x}}} \left\|\mathbf{p} - \mathbf{q}\right\|_{\mathbf{x}} \right|\\
&\leq \frac{2}{\left\|\mathbf{p}\right\|_{\mathbf{x}}}\left\|\mathbf{p} - \mathbf{q}\right\|_{\mathbf{x}} .
\end{align*}
Thus for arbitrary $\varepsilon > 0$, there exists $\delta_1 > 0$ such that for $\|\mathbf{p} - \mathbf{q}\| < \delta_1$,
$$ \left| \sum_{i=j_1}^{j_2} w_i^*\left(\mathbf{p}\right) -  \sum_{j_1}^{j_2} w_i^*\left(\mathbf{q}\right) \right|
 <  \varepsilon $$
for all $1 \leq j_1 \leq j_2 \leq D_{\mathbf{x}}$.

Also there exists $\delta_2 > 0$ such that for $\mathbf{q} \in \Delta^*$ where $\left\|\mathbf{p} - \mathbf{q}\right\|_{\mathbf{x}} < \delta_2$,
$$ \left| \chi\left(\mathbf{p}\right) - \chi\left(\mathbf{q}\right) \right| < 
\left| \vartheta\left(\mathbf{p}\right) - \vartheta\left(\mathbf{q}\right) \right| $$
where $q_r > 0$.  Then,
\begin{align*}
\left|\vartheta\left(\mathbf{p}\right)\right.&-\left.\vartheta\left(\mathbf{q}\right)\right| = \\
&\left|\frac{\alpha - \sum_{i=1}^{r-1} w_i^*\left(\mathbf{p}\right)}{w_r^*\left(\mathbf{p}\right)} -
       \frac{\alpha - \sum_{i=1}^{r-1} w_i^*\left(\mathbf{q}\right)}{w_r^*\left(\mathbf{p}\right)} +\right.\\
       &\quad\quad\quad\quad\quad\quad\quad\quad\,\,\,\left.
\frac{\alpha - \sum_{i=1}^{r-1} w_i^*\left(\mathbf{q}\right)}{w_r^*\left(\mathbf{p}\right)} -
       \frac{\alpha - \sum_{i=1}^{r-1} w_i^*\left(\mathbf{q}\right)}{w_r^*\left(\mathbf{q}\right)} \right| \\
&=\frac{1}{w_r^*\left(\mathbf{p}\right)}\left|\sum_{i=1}^{r-1} \left( w_i^*\left(\mathbf{p}\right) -
w_i^*\left(\mathbf{q}\right)\right) + \vartheta\left(\mathbf{q}\right) \left(w_r^*\left(\mathbf{p}\right) -
w_r^*\left(\mathbf{q}\right)\right)\right| \\
&\leq \frac{1}{w_r^*\left(\mathbf{p}\right)}\left[\sum_{i=1}^{r-1} \left| w_i^*\left(\mathbf{p}\right) -
w_i^*\left(\mathbf{q}\right)\right| + \left|w_r^*\left(\mathbf{p}\right) - w_r^*\left(\mathbf{q}\right)\right|\right].
\end{align*}
Therefore, if 
$$\left\|\mathbf{p}-\mathbf{q}\right\|_{\mathbf{x}} < \Min\left\{\frac{p_r \varepsilon \sum_{\pi \in C_r} g_0\left(\mathbf{x}\right)}{4},
\delta_2\right\} $$
then
$$ \left| \sum_{i=j_1}^{j_2} w_i^*\left(\mathbf{p}\right) -  \sum_{j_1}^{j_2} w_i^*\left(\mathbf{q}\right) \right|
 < \frac{p_r \varepsilon \sum_{\pi \in C_r} g_0\left(\mathbf{x}\right)}{2 \left\|\mathbf{p}\right\|_{\mathbf{x}}}, $$
for all $1 \leq j_1 \leq j_2 \leq D_x$.  Consequently,
$$ \left| \chi\left(\mathbf{p}\right) - \chi\left(\mathbf{q}\right) \right| < \varepsilon $$
and the $\chi$ is shown to be continuous in the modified $L^1$-norm.
\end{proof}

\begin{theorem}
For $\mathbf{x} \in \Omega_r$, if all permutations of $\mathfrak{S}_n$ are equally likely to be selected in $S$, a sample with replacement of size $s$, then:
$$ \lim_{s \to \infty} {\hat \varphi}\left(\mathbf{x};\mathbf{p}\right) = \varphi \left(\mathbf{x}\right).$$
That is, the approximate generalized permutation test, ${\hat \varphi}$, converges pointwise to the most powerful generalized permutation test.
\end{theorem}
\begin{proof}
If each permutation is equally likely to be selected in $S$, then the values of $p_i$ are identical and $w^*_i\left(\mathbf{p}\right) = w_i$ for all $i=1$,$\ldots$,$D_x$.  Therefore $\varphi\left(\mathbf{x}\right)=\chi\left(\mathbf{p}\right)$.  In particular, $p_i$ is non-zero and therefore $\chi$ is continuous at $\mathbf{p}$.  For arbitrary $\varepsilon > 0$, let $\delta$ be such that
$$ \left| \chi\left(\mathbf{p}\right) - \chi\left(\mathbf{q}\right) \right| < \frac{\varepsilon}{2} $$
whenever $\left\|\mathbf{p} - \frac{1}{s}\mathbf{k}\right\|_{\mathbf{x}} < \delta$ for $\mathbf{k} \in \Delta_s$.

Note that:
\begin{align*}
\left| {\hat \varphi}\left(\mathbf{x};\mathbf{p}\right) - \chi\left(\mathbf{p}\right) \right| &=
\left| \sum_{\mathbf{k} \in \Delta_s} \chi\left(\frac{\mathbf{k}}{s}\right)M_{\mathbf{k}}\left(\mathbf{p}\right) -
\chi\left(\mathbf{p}\right)\sum_{\mathbf{k} \in \Delta_s} M_{\mathbf{k}}\left(\mathbf{p}\right) \right| \\
&\leq \sum_{\mathbf{k} \in \Delta_s} \left| \chi\left(\frac{\mathbf{k}}{s}\right) - \chi\left(\mathbf{p}\right)
\right| M_{\mathbf{k}}\left(\mathbf{p}\right).
\end{align*}
Therefore, breaking this summation into two parts, where:
\begin{align*}
\Delta_+ &= \left\{ \mathbf{k} \in \Delta_s : \left\|\mathbf{p} - \frac{1}{s}\mathbf{k}\right\|_{\mathbf{x}} \geq
 \delta \right\} \text{ and}\\
\Delta_- &= \left\{ \mathbf{k} \in \Delta_s : \left\|\mathbf{p} - \frac{1}{s}\mathbf{k}\right\|_{\mathbf{x}} <
 \delta \right\},
\end{align*}
results in:
\begin{align*}
\left| {\hat \varphi}\left(\mathbf{x};\mathbf{p}\right) - \chi\left(\mathbf{p}\right) \right| &\leq
\sum_{\mathbf{k} \in \Delta_-} \left| \chi\left(\frac{\mathbf{k}}{s}\right) - \chi\left(\mathbf{p}\right)
\right| M_{\mathbf{k}}\left(\mathbf{p}\right) +\\
&\sum_{\mathbf{k} \in \Delta_+} \left| \chi\left(\frac{\mathbf{k}}{s}\right) - \chi\left(\mathbf{p}\right)
\right| M_{\mathbf{k}}\left(\mathbf{p}\right).
\end{align*}
By continuity, the first summand is less than $\frac{\varepsilon}{2}$.  Furthermore the absolute difference in values of $\chi$ is bounded by one:
$\left| \chi\left(\mathbf{p}\right) - \chi\left(\mathbf{q}\right)\right| \leq 1$, therefore:
$$ \left| {\hat \varphi}\left(\mathbf{x};\mathbf{p}\right) - \chi\left(\mathbf{p}\right) \right| \leq
\frac{\varepsilon}{2} + \sum_{\mathbf{k} \in \Delta_+} M_{\mathbf{k}}\left(\mathbf{p}\right). $$
Consequently, by the concentration of measures inequality for $M_{\mathbf{k}}\left(\mathbf{p}\right)$, there exists $s$ sufficiently large that:
$$\sum_{\mathbf{k} \in \Delta_+} M_{\mathbf{k}}\left(\mathbf{p}\right) < \frac{\varepsilon}{2}$$
and convergence follows.
\end{proof}

\section{Estimation of Significance}

\subsection{Direct Approach}

Consider the case where the goal of the analysis is to estimate a $p$-value for a particular sample within the context of a given null hypothesis, deciding whether or not to reject that null hypothesis.  In other words, the goal of the analysis is to determine the maximum significance level $\alpha$ for which a given null hypotheses would be rejected were the test performed, and thereby quantify the ``extremeness'' of the sample under a given null model.  For the generalized permutation test, formally, these approaches differ in a number of subtle aspects, most notably in the treatment of sample points that have the same likelihood ratio.  Given positive measurable functions $g_0$ and $g_1$ representing the probability densities of null and alternative hypotheses respectively, significance level is defined as follows.

\begin{definition} For $\mathbf{x} \in \Omega_r$ with non-zero $h$, the significance of the most powerful generalized permutation test is defined to be:
$$ \alpha = \sum_{i=1}^{r-1} w_i = \frac{\sum_{i=1}^{r-1} \sum_{\pi \in C_i}
 g_0\left(\pi \mathbf{x}\right)}{n! h\left(\mathbf{x}\right)}$$
where $\mathfrak{S}_n$ is ordered following (\ref{eq:lr-chain}).
\end{definition}

Analogous to the case of hypothesis testing using the most powerful generalized permutation test, calculating the significance requires computing the likelihood ratio of $g_1$ and $g_0$ for all $n!$ permutations of any sample point $\mathbf{x}$.  Given that this is, in most practical case, not feasible, approximate approaches are required.  The most direct Monte Carlo approximation approach consists of considering a sample of the full group of permutation and calculating likelihood ratios within that sample: the resulting estimate of significance being the fraction of the total number of sample permutations with likelihood ratios larger than the likelihood ratio of the original sample point.

\begin{theorem}
For probability densities $g_0$ and $g_1$ representing null and alternative hypotheses respectively, let $S$ be a sample with replacement from $\mathfrak{S}_n$ of size $s$.  Assume $h\left(\mathbf{x}\right)$ is non-zero.  Let $\Sigma_U$ denote the sum of the null hypothesis values over the sample of permutations $S$ with likelihood ratio less than $l\left(\mathbf{x}\right) = g_1\left( \mathbf{x} \right) / g_0\left( \mathbf{x} \right)$ and let $\Sigma_V$ denote the sum of all of the null hypothesis values. If each permutation in $\mathfrak{S}_n$ is equally likely to be selected for $S$, then:
\begin{align*}
&\Prob\left[ \left| \frac{\Sigma_U}{\Sigma_V}  - \alpha \right| > \varepsilon \right] \leq
\Exp\left({-\frac{1}{2}\frac{s \varepsilon^2 {\bar v}^2}{\tau_1 + c \varepsilon /3}}\right) +
\Exp\left({-\frac{1}{2}\frac{s \varepsilon^2 {\bar v}^2}{\tau_1^\prime + c^\prime \varepsilon /3}}\right) 
\text{ and}\\
&\Prob\left[ \left| \frac{\Sigma_U}{\Sigma_V}  - \alpha \right| > \varepsilon \right] \leq
\Exp\left(-2 s \varepsilon^2 \tau_2^2\right) + \Exp\left(-2 s \varepsilon^2 {\tau_2^\prime}^2\right)
\end{align*}
where $\alpha$ is the significance of the most powerful generalized permutation test and $0 < \varepsilon < \Min\left\{\alpha, 1-\alpha\right\}$.  In particular, as the size of the sample increases towards infinity, the significance of the approximate generalized permutation test approaches the significance of the most powerful generalized permutation test.
\end{theorem}
\begin{proof}
Consider $S$, a sample with replacement from $\mathfrak{S}_n$ of size $s$.  Assume that each permutation in $\mathfrak{S}_n$ is equally likely to be selected for $S$.  For $\mathbf{x} \in \Omega_r$ and $\mathfrak{S}_n$ ordered according to (\ref{eq:lr-chain}), consider the vector-valued, discrete random variable $\left(U,V\right)$ where $\Prob\left[\left(U,V\right)=\left(u,v\right)\right]=0$ for all $u,v$ except:
\begin{align*}
&\Prob\left[\left(U,V\right)=\left(\sum_{\pi \in C_i} g_0\left(\pi \mathbf{x}\right),
\sum_{\pi \in C_i} g_0\left(\pi \mathbf{x}\right)\right)\right]
 = \frac{1}{D_x} \text{ for } i < r \leq D_x\\
&\Prob\left[\left(U,V\right)=\left(0,\sum_{\pi \in C_i} g_0\left(\pi \mathbf{x}\right)\right)\right] =
 \frac{1}{n!} \text{ for } r \leq i \leq D_x.
\end{align*}

Letting:
\begin{align*}
{\bar u}&= \frac{1}{D_x}\sum_{i=1}^{r-1} \sum_{\pi \in C_i} g_0 \left(\pi \mathbf{x}\right) \text{ and } \\
{\bar v}&= \frac{1}{D_x} n! h \left(\mathbf{x}\right),
\end{align*}
the expected value of $\left(U,V\right)$, $\Ex\left[\left(U,V\right)\right] = \left({\bar u},{\bar v}\right)$ and the significance of the generalized permutation test, $\alpha$, is ${\bar u}/{\bar v}$. 

Consider $\left(U_i,V_i\right)$ for $i=1$,$\ldots$,$s$ independent random variables distributed identically to $\left(U,V\right)$.  Since:
$$0 \leq \sum_{i=1}^s U_i \leq \sum_{i=1}^s V_i \text{ and } \sum_{i=1}^s V_i > 0,$$
it follows that:
$$\Prob\left[ \frac{\sum_{i=1}^s U_i}{\sum_{i=1}^s V_i} - \alpha > 1 - \alpha \right] =
\Prob\left[ \frac{\sum_{i=1}^s U_i}{\sum_{i=1}^s V_i} - \alpha < -\alpha \right] = 0.$$
However, for $0 < \varepsilon < 1 - \alpha$, define:
\begin{align*}
Z_i &= U_i {\bar v} - \left({\bar u} + \varepsilon {\bar v}\right)V_i \\
Z_i^\prime &= U_i {\bar v} - \left({\bar u} - \varepsilon {\bar v}\right)V_i,
\end{align*}
then
$$ \frac{\sum_{i=1}^s U_i}{\sum_{i=1}^s V_i} - \frac{{\bar u}}{{\bar v}} > \varepsilon $$
if and only if $\sum_{i=1}^s Z_i > 0$.  Similarly:
$$ \frac{\sum_{i=1}^s U_i}{\sum_{i=1}^s V_i} - \frac{{\bar u}}{{\bar v}} <  -\varepsilon $$
if and only if $\sum_{i=1}^s Z_i^\prime < 0$.

Letting:
\begin{align*}
&a = \Min \left\{\pi_i g_0\left(\mathbf{x}\right):i=1,\ldots,n!\right\} \text{ and } \\
&b = \Max \left\{\pi_i g_0\left(\mathbf{x}\right):i=1,\ldots,n!\right\},
\end{align*}
then:
\begin{align*}
&\Ex\left[Z_i\right] = -\varepsilon {\bar v}^2, \\
&\Min\left(Z_i - \Ex\left[Z_i\right]\right) = - \left({\bar u} + \varepsilon {\bar v}\right)b + \varepsilon{\bar v}^2,
 \text{ and }\\
&\Max\left(Z_i - \Ex\left[Z_i\right]\right) = {\bar v}b - \left({\bar u} + \varepsilon {\bar v}\right)a +
 \varepsilon{\bar v}^2.
\end{align*}

If in addition $0 < \varepsilon < \alpha$, then:
\begin{align*}
&\Ex\left[Z_i^\prime\right] = \varepsilon {\bar v}^2, \\
&\Min\left(Z_i^\prime - \Ex\left[Z_i^\prime\right]\right) = - \left({\bar u} - \varepsilon {\bar v}\right)b -
 \varepsilon{\bar v}^2,
 \text{ and }\\
&\Max\left(Z_i^\prime - \Ex\left[Z_i^\prime\right]\right) = {\bar v}b - \left({\bar u} - \varepsilon {\bar v}\right)a -
 \varepsilon{\bar v}^2.
\end{align*}
By Berstein's Inequality (\cite{Uspensky}, p. 205):
\begin{align*}
&\Prob\left[ \sum_{i=0}^s \left(Z_i - \Ex\left[Z_i\right]\right) > t\right]
 \leq \Exp\left({-\frac{t^2}{2 s \Var{Z_i}+2 c t/3}}\right) \text{ and}\\
&\Prob\left[ \sum_{i=0}^s \left(Z_i^\prime - \Ex\left[Z_i^\prime\right]\right) < -t\right]
 \leq \Exp\left({-\frac{t^2}{2 s \Var{Z_i^\prime}+2c^\prime t/3}}\right)
\end{align*}
for $t>0$ where:
\begin{align*}
&c=\Max\left\{\left({\bar u} - \varepsilon {\bar v}\right)b - \varepsilon{\bar v}^2 ,
{\bar v}b - \left({\bar u} + \varepsilon {\bar v}\right)a + \varepsilon{\bar v}^2\right\} \text{ and}\\
&c^\prime=\Max\left\{\left({\bar u} - \varepsilon {\bar v}\right)b + \varepsilon{\bar v}^2,
{\bar v}b - \left({\bar u} - \varepsilon {\bar v}\right)a - \varepsilon{\bar v}^2\right\}.
\end{align*}
Setting $t = -s \Ex\left[Z_i\right] = s \varepsilon {\bar v}^2$,
$$ \Prob\left[ \sum_{i=0}^s Z_i  > 0\right]
\leq \Exp\left({-\frac{s^2 \varepsilon^2 {\bar v}^4}{2 s \Var{Z_i}+2 c s {\bar v}^2/3}}\right), $$
or equivalently:
$$ \Prob\left[ \frac{\sum_{i=0}^s U_i}{\sum_{i=0}^s V_i} - \alpha  > \varepsilon\right]
\leq \Exp\left({-\frac{1}{2}\frac{s \varepsilon^2 {\bar v}^2}{\tau_1 + c \varepsilon /3}}\right)$$
for $\tau_1 = \Var{Z_i} / {\bar v}^2$.
Similarly for $t = s \Ex\left[Z_i^\prime\right] = s \varepsilon {\bar v}^2$,
$$ \Prob\left[ \sum_{i=0}^s Z_i^\prime  < 0\right]
\leq \Exp\left({-\frac{s^2 \varepsilon^2 {\bar v}^4}{2 s \Var{Z_i^\prime}+2 c^\prime s {\bar v}^2/3}}\right), $$
or equivalently:
$$ \Prob\left[ \frac{\sum_{i=0}^s U_i}{\sum_{i=0}^s V_i} - \alpha  < -\varepsilon\right]
\leq \Exp\left({-\frac{1}{2}\frac{s \varepsilon^2 {\bar v}^2}{\tau^\prime + c^\prime \varepsilon /3}}\right)$$
for $\tau_1^\prime = \Var{Z_i^\prime} / {\bar v}^2$.

The second inequality is obtained by applying Hoeffding's inequality (\cite{Hoeffding3}) in place of Bernstein's.  Letting: 
\begin{align*}
&\tau_2 = \frac{{\bar v}^2}{{\bar u}\left(b-a\right)+{\bar v}\left(\left(1+\varepsilon\right)b - \varepsilon a\right)}
\text{ and } \\
&\tau_2^\prime = \frac{{\bar v}^2}{{\bar u}\left(b-a\right)+{\bar v}\left(\left(1-\varepsilon\right)b -
 \varepsilon a\right)},
\end{align*}
results in the inequalities:
\begin{align*}
&\Prob\left[ \frac{\Sigma_U}{\Sigma_V}  - \alpha > \varepsilon \right] \leq
\Exp\left(-2 s \varepsilon^2 \tau_2^2\right) \text{ and } \\
&\Prob\left[ \frac{\Sigma_U}{\Sigma_V}  - \alpha < -\varepsilon \right] \leq
\Exp\left(-2 s \varepsilon^2 {\tau_2^\prime}^2\right)
\end{align*}
\end{proof}

\subsection{Indirect Approach}

Although the estimator derived from direct approach results in a bound on the absolute deviation from the significance level of the most powerful generalized permutation test, a more analytically tractable approach can be obtained by introducing unequal probability of selecting each permutation in the sample.  Consider $S$, a sample with replacement from $\mathfrak{S}_n$ of size $s$.  Let $\mathbf{x} \in \Omega_r$ and $\mathfrak{S}_n$ be ordered according to (\ref{eq:lr-chain}).  Assume that each of the $n!/D_x$ permutations in each class $C_i \subseteq \mathfrak{S}_n$ have the same probability; however, the overall probability of selecting a permutation from class $C_i$ is given by $p_i$.  That is, for $\pi \in C_i$, the probability of selecting $\pi$ for $S$ is $p_i D_x / n!$.

\begin{theorem}
For probability densities $g_0$ and $g_1$ representing null and alternative hypotheses respectively, let $S$ be a sample with replacement from $\mathfrak{S}_n$ of size $s$.  Assume $h\left(\mathbf{x}\right)$ is non-zero.  Let $\Sigma_U$ denote the fraction of null hypothesis values in the sample of permutations $S$ with likelihood ratio greater than
$l\left(\mathbf{x}\right) = g_1\left( \mathbf{x} \right) / g_0\left( \mathbf{x} \right)$. If the probability of selecting a permutation in $C_i \subseteq \mathfrak{S}_n$ is:
$$p_i= \frac{\sum_{\pi \in C_i} g_0\left(\pi \mathbf{x}\right)}{n! h\left(\mathbf{x}\right)},$$
then $\Sigma_U$ is an unbiased estimator of the significance of the most powerful generalized permutation test and:
\begin{align*}
&\Prob\left[ \left| \frac{1}{s}\sum_{i=1}^s U_i - \sum_{i=1}^r p_i \right| > \varepsilon\right] \leq
2 \Exp{\left(-2 s \varepsilon^2\right)} \text{ and} \\
&\Prob\left[ \left| \frac{1}{s}\sum_{i=1}^s U_i - \sum_{i=1}^r p_i \right| > \varepsilon\right] \leq
2 \Exp{\left(-\frac{1}{2}\frac{s \varepsilon^2}{\sigma^2+c \varepsilon /3}\right)}\\
\end{align*}
where $\alpha$ is the significance of the most powerful generalized permutation test and $\varepsilon > 0$.
\end{theorem}
\begin{proof}
Define the Bernoulli random variable $U$ by:
\begin{align*}
&\Prob\left[U=1\right] = \sum_{i=1}^{r-1} p_i \\
&\Prob\left[U=0\right] = 1 - \sum_{i=1}^{r-1} p_i = \sum_{i=r}^{D_x} p_i.
\end{align*}
Consider $U_1$,$U_2$,$\ldots$,$U_s$ independent and distributed identically to $U$. By Hoeffding's (\cite{Hoeffding3}) and Bernstein's (\cite{Uspensky}) inequalities:
\begin{align*}
&\Prob\left[ \left| \frac{1}{s}\sum_{i=1}^s U_i - \sum_{i=1}^r p_i \right| > \varepsilon\right] \leq
2 \Exp{\left(-2 s \varepsilon^2\right)} \text{ and} \\
&\Prob\left[ \left| \frac{1}{s}\sum_{i=1}^s U_i - \sum_{i=1}^r p_i \right| > \varepsilon\right] \leq
2 \Exp{\left(-\frac{1}{2}\frac{s \varepsilon^2}{\sigma^2+c \varepsilon /3}\right)}\\
\end{align*}
where $c=\Max\left\{1 - \sum_{i=1}^{r-1} p_i , \sum_{i=1}^{r-1} p_i\right\}$ and 
$\sigma^2 = \left(1 - \sum_{i=1}^{r-1} p_i\right)\left(\sum_{i=1}^{r-1} p_i\right)$ for all $\varepsilon > 0$.
\end{proof}

The difficulty in applying the theorem is that in order to construct the discrete probability distribution on the permutations it is necessary to compute $n!$ values of $g_0$.  In order to avoid this excessive computation, consider the following approach to specifying the probability distribution on $\mathfrak{S}_n$.

For $\mathbf{x} \in \Omega_r$, let $\mathbf{x} \in \pi_0 T$ for some $\pi_0 \in \mathfrak{S}_n$ and let $\mathfrak{S}_n$ be ordered according to:
\begin{equation}\label{eq:lr-chain2}
l_1\left(\mathbf{x}\right) > l_2\left(\mathbf{x}\right) > l\left(\mathbf{x}\right) = l_r\left(\mathbf{x}\right) >
 \ldots > l_D\left(\mathbf{x}\right).
\end{equation}
Let ${\hat g}_0$ be an approximation, parametric or non-parametric, to $g_0$ constructed from the values in $g_0\left(\mathbf{x}\right)$.
For any randomly selected vector $\mathbf{x}^*$, there exists some permutation $\pi$ such that $\pi \mathbf{x}^* \in \pi_0 T$.  Explicitly, if $\mathbf{x}^* \in \pi^\prime T$, then for $\pi = \pi_0 {\pi^\prime}^{-1}$, $\pi \mathbf{x}^* \in \pi_0 T$.  That is, $\pi$ is the permutation that moves $\mathbf{x}^*$ into the asymmetric unit containing $\mathbf{x}$.  Therefore, each random selection $\mathbf{x}^*$ induces a particular permutation.

The probability of selecting any particular permutation depends on the probability of selecting a point from any given asymmetric unit of $\Omega$.  For any $\pi_i$ in (\ref{eq:lr-chain2}), $\pi_i$ will be selected if $\pi_i = \pi_0 {\pi^\prime}^{-1}$ or equivalently $\pi^\prime = {\pi_i}^{-1} \pi_0$.  Thus $\pi_i$ will be selected if $\mathbf{x}^* \in {\pi_i}^{-1} \pi_0 T$.  Consequently this selection process will result in the following probability distribution on $\mathfrak{S}_n$ where the probability, $p_i$, of selecting a given permutation $\pi_i$ is given by:
\begin{equation}\label{eq:id-prob}
p_i = \int_{\pi_i^{-1} \pi_0 T} {\hat g}_0\,d\mu.
\end{equation}
for $\mathbf{x} \in \Omega_r$ and $\mathbf{x} \in \pi_0 T$.

The properties of $\Sigma_U$ as an estimator based on the probability density (\ref{eq:id-prob}) naturally depend on the specifics of the approximation ${\hat g}_0$.  In particular:
\begin{align*}
&\Ex\left[\Sigma_U\right] = \sum_{i=1}^{r-1} \int_{\pi_i^{-1} \pi_0 T} {\hat g}_0\,d\mu \text{ and}\\
&\Var\left[\Sigma_U\right] = \left(\sum_{i=1}^{r-1} \int_{\pi_i^{-1} \pi_0 T} {\hat g}_0\,d\mu\right)
\left(\sum_{i=r}^{D} \int_{\pi_i^{-1} \pi_0 T} {\hat g}_0\,d\mu\right)
\end{align*}
for $\mathbf{x} \in \Omega_r$ and $\mathbf{x} \in \pi_0 T$.

\subsection*{Acknowledgement}
We would like to thank Leonard Scott, Wei Sun, and Yufeng Liu for reviewing the manuscript and providing valuable improvements.


\end{document}